\documentclass[journal]{IEEEtran}
\usepackage{graphicx,amsmath,amsfonts,amscd,amssymb,bm,cite,epsfig,epsf,url,color}
\usepackage{fullpage} \usepackage[small,bf]{caption}
\newtheorem{theorem}{Theorem}
\newtheorem{lemma}[theorem]{Lemma}

\newtheorem{definition}[theorem]{Definition}

\newcommand{\R}{\mathbb{R}}

\newcommand{\<}{\langle}
\renewcommand{\>}{\rangle}

\newcommand{\goto}{\rightarrow}
\newcommand{\sgn}{\textrm{sgn}}

\newcommand{\E}{\operatorname{\mathbb{E}}}

\newcommand{\PT}{\mathcal{P}_T}
\newcommand{\PO}{\mathcal{P}_\Omega}
\newcommand{\POc}{\mathcal{P}_{\Omega^c}}
\newcommand{\PTp}{\mathcal{P}_{T^\perp}}
\newcommand{\cA}{\mathcal{A}}

\newcommand{\vct}[1]{{#1}}
\newcommand{\mtx}[1]{{#1}}

\newcommand{\rank}{\operatorname{rank}}

\newcommand{\trace}{\operatorname{trace}}

\newcommand{\OpId}{\mathcal{I}}

\numberwithin{equation}{section}

\def \endprf{\hfill {\vrule height6pt width6pt depth0pt}\medskip}
\newenvironment{proof}{\noindent {\bf Proof} }{\endprf\par}

\title{Matrix Completion with Noise}

\author{Emmanuel J. Cand\`es  and Yaniv Plan
Applied and Computational Mathematics,
Caltech, Pasadena, CA 91125}

\date{\today}

\begin{document}

\maketitle

\vspace{-0.3in}

\begin{abstract}
  On the heels of compressed sensing, a remarkable new field has very
  recently emerged. This field addresses a broad range of problems of
  significant practical interest, namely, the recovery of a data
  matrix from what appears to be incomplete, and perhaps even
  corrupted, information. In its simplest form, the problem is to
  recover a matrix from a small sample of its entries, and comes up in
  many areas of science and engineering including collaborative
  filtering, machine learning, control, remote sensing, and computer
  vision to name a few. 

  This paper surveys the novel literature on matrix completion, which
  shows that under some suitable conditions, one can recover an
  unknown low-rank matrix from a nearly minimal set of entries by
  solving a simple convex optimization problem, namely, nuclear-norm
  minimization subject to data constraints. Further, this paper
  introduces novel results showing that matrix completion is provably
  accurate even when the few observed entries are corrupted with a
  small amount of noise. A typical result is that one can recover an
  unknown $n \times n$ matrix of low rank $r$ from just about $nr
  \log^2 n$ noisy samples with an error which is proportional to the
  noise level.  We present numerical results which complement our
  quantitative analysis and show that, in practice, nuclear norm
  minimization accurately fills in the many missing entries of large
  low-rank matrices from just a few noisy samples. Some analogies
  between matrix completion and compressed sensing are discussed
  throughout.
\end{abstract}


\section{Introduction}
\label{sec:introduction}

Imagine that we only observe a few samples of a signal. Is it possible
to reconstruct this signal exactly or at least accurately?  For
example, suppose we observe a few entries of a vector $x \in \R^n$,
which we can think of as a digital signal or image. Can we recover the
large fraction of entries---of pixels if you will---that we have not
seen?  In general, everybody would agree that this is
impossible. However, if the signal is known to be sparse in the
Fourier domain and, by extension, in an incoherent domain, then
accurate---and even exact---recovery is possible by $\ell_1$
minimization \cite{CRT06}, see also \cite{GilbertStraussII} for other
algorithms, \cite{DonohoL1L0,DonohoTanner} for other types of
measurements, and \cite{Vetterli} for different ideas.  This
revelation is at the root of the rapidly developing field of
compressed sensing, and is already changing the way engineers think
about data acquisition, hence this special issue and others, see
\cite{SPM} for example. Concretely, if a signal has a sparse frequency
spectrum and we only have information about a few time or space
samples, then one can invoke linear programming to interpolate the
signal exactly. One can of course exchange time (or space) and
frequency, and recover sparse signals from just a few of their Fourier
coefficients as well.

Imagine now that we only observe a few entries of a data matrix. Then
is it possible to accurately---or even exactly---guess the entries
that we have not seen?  For example, suppose we observe a few movie
ratings from a large data matrix in which rows are users and columns
are movies (we can only observe a few ratings because each user is
typically rating a few movies as opposed to the tens of thousands of
movies which are available). Can we predict the rating a user would
hypothetically assign to a movie he/she has not seen? In general,
everybody would agree that recovering a data matrix from a subset of
its entries is impossible. However, if the unknown matrix is known to
have low rank or approximately low rank, then accurate and even exact
recovery is possible by nuclear norm minimization
\cite{CR08,CT09}. This revelation, which to some extent is inspired by
the great body of work in compressed sensing, is the subject of this
paper.

From now on, we will refer to the problem of inferring the many
missing entries as the {\em matrix completion} problem. By extension,
inferring a matrix from just a few linear functionals will be called
the {\em the low-rank matrix recovery} problem. Now just as sparse
signal recovery is arguably of paramount importance these days, we do
believe that matrix completion and, in general, low-rank matrix
recovery is just as important, and will become increasingly studied in
years to come. For now, we give a few examples of applications in which
these problems do come up.

\begin{itemize}
\item {\em Collaborative filtering.}  In a few words, collaborative
  filtering is the task of making automatic predictions about the
  interests of a user by collecting taste information from many users
  \cite{Goldberg92}. Perhaps the most well-known implementation of
  collaborating filtering is the Netflix recommendation system alluded
  to earlier, which seeks to make rating predictions about unseen
  movies. This is a matrix completion problem in which the unknown
  full matrix has approximately low rank because only a few factors
  typically contribute to an individual's tastes or preferences. In
  the new economy, companies are interested predicting musical
  preferences (Apple Inc.), literary preferences (Amazon, Barnes and
  Noble) and many other such things. 

\item {\em System identification.} In control, one would like to fit a
  discrete-time linear time-invariant state-space model
\begin{align*}
x(t + 1) & = Ax(t) + Bu(t),\\
y(t) & =  C x(t) + Du(t)
\end{align*}
to a sequence of inputs $u(t) \in \R^m$ and outputs $y(t) \in R^p$, $t
= 0, \ldots, N$. The vector $x(t) \in \R^n$ is the state of the system
at time $t$, and $n$ is the order of the system model.  From the
input/output pair $\{(u(t), y(t)) : t = 0, \ldots N\}$, one would like
to recover the dimension of the state vector $n$ (the model order),
and the dynamics of the system, i.e.~the matrices $A$, $B$, $C$, $D$,
and the initial state $x(0)$. This problem can be cast as a low-rank
matrix recovery problem, see \cite{VandenbergheNuc} and references
therein.

\item {\em Global positioning.} Finding the global positioning of
  points in Euclidean space from a local or partial set of pairwise
  distances is a problem in geometry that emerges naturally in sensor
  networks \cite{TohYe,Singer1,Singer2}. For example, because of power
  constraints, sensors may only be able to construct reliable distance
  estimates from their immediate neighbors. From these estimates, we
  can form a partially observed distance matrix, and the problem is to
  infer all the pairwise distances from just a few observed ones so
  that locations of the sensors can be reliably estimated.  This
  reduces to a matrix completion problem where the unknown matrix is
  of rank two if the sensors are located in the plane, and three if
  they are located are in space.

\item {\em Remote sensing.} The MUSIC algorithm \cite{MUSIC} is
  frequently used to determine the direction of arrival of incident
  signals in a coherent radio-frequency environment. In a typical
  application, incoming signals are being recorded at various sensor
  locations, and this algorithm operates by extracting the directions
  of wave arrivals from the covariance matrix obtained by computing
  the correlations of the signals received at all sensor pairs. In
  remote sensing applications, one may not be able to estimate or
  transmit all correlations because of power constraints
  \cite{Brian}. In this case, we would like to infer a full covariance
  matrix from just a few observed partial correlations. This is a
  matrix completion problem in which the unknown signal covariance
  matrix has low rank since it is equal to the number of incident
  waves, which is usually much smaller than the number of sensors.
\end{itemize}

There are of course many other examples including the
structure-from-motion problem \cite{Tomasi,ChenSuter} in computer
vision, multi-class learning in data analysis \cite{Argyriou07,
  Amit07}, and so on.

This paper investigates whether or not one can recover low rank
matrices from fewer entries, and if so, how and how well. In Section
\ref{sec:exact}, we will study the noiseless problem in which the
observed entries are precisely those of the unknown matrix. Section
\ref{sec:stable} examines the more common situation in which the few
available entries are corrupted with noise. We complement our study
with a few numerical experiments demonstrating the empirical
performance of our methods in Section \ref{sec:numerical} and conclude
with a short discussion (Section \ref{sec:discussion}).

Before we begin, it is best to provide a brief summary of the
notations used throughout the paper. We shall use three norms of a
matrix $X \in \R^{n_1 \times n_2}$ with singular values
$\{\sigma_k\}$. The \emph{spectral norm} is denoted by $\|\mtx{X}\|$
and is the largest singular value. The Euclidean inner product between
two matrices is defined by the formula $\<\mtx{X}, \mtx{Y}\> :=
\trace(\mtx{X}^* \mtx{Y})$, and the corresponding Euclidean norm is
called the \emph{Frobenius norm} and denoted by $\|X\|_F$ (note that
this is the $\ell_2$ norm of the vector of singular values). The
\emph{nuclear norm} is denoted by $\|\mtx{X}\|_* := \sum_k \sigma_k$
and is the sum of singular values (the $\ell_1$ norm of the vector
$\{\sigma_k\}$). As is standard, $X \succeq Y$ means that $X - Y$ is
positive semidefinite.

Further, we will also manipulate linear transformation which acts on
the space $\R^{n_1 \times n_2}$, and we will use calligraphic letters
for these operators as in ${\mathcal A}(\mtx{X})$. In particular, the
identity operator on this space will be denoted by $\OpId: \R^{n_1
  \times n_2} \to \R^{n_1 \times n_2}$. We use the same convention as
above, and $\mathcal{A} \succeq \OpId$ means that $\mathcal{A} -
\OpId$ (seen as a big matrix) is positive semidefinite. 

We use the usual asymptotic notation, for instance writing $O(M)$ to
denote a quantity bounded in magnitude by $C M$ for some absolute
constant $C > 0$.

\section{Exact Matrix Completion}
\label{sec:exact}

From now on, $M \in \R^{n_1 \times n_2}$ is a matrix we would like to
know as precisely as possible. However, the only information available
about M is a sampled set of entries $M_{ij}$, $(i,j) \in \Omega$,
where $\Omega$ is a subset of the complete set of entries $[n_1]
\times [n_2]$. (Here and in the sequel, $[n]$ denotes the list
$\{1,\ldots,n\}$.) It will be convenient to summarize the information
available via $\PO(M)$, where the sampling operator $\PO: \R^{n_1
  \times n_2} \to \R^{n_1 \times n_2}$ is defined by
\[
[\PO(X)]_{ij} = \begin{cases} X_{ij}, & (i,j) \in \Omega,\\
  0, & \text{otherwise}. \end{cases}
\]
Thus, the question is whether it is possible to recover our matrix
only from the information $\PO(\mtx{M})$.  We will assume that the
entries are selected at random without replacement as to avoid trivial
situations in which a row or a column is unsampled, since matrix
completion is clearly impossible in such cases. (If we have no data
about a specific user, how can we guess his/her preferences? If we
have no distance estimates about a specific sensor, how can we guess
its distances to all the sensors?)

Even with the information that the unknown matrix $M$ has low rank, this
problem may be severely ill posed. Here is an example that shows why:
let $x$ be a vector in $\R^n$ and consider the $n \times n$ rank-1
matrix
\[
\mtx{M} = \vct{e}_1 \vct{x}^* = \begin{bmatrix}
 x_1 & x_2 & x_3 & \cdots & x_{n-1} & x_n \\
 0 & 0 & 0 & \cdots & 0 & 0 \\
 0 & 0 & 0 & \cdots & 0 & 0 \\
 \vdots & \vdots & \vdots &  \vdots & \vdots & \vdots \\
 0 & 0 & 0 & \cdots & 0 & 0
\end{bmatrix},
\]
where $e_1$ is the first vector in the canonical basis of $\R^n$.
Clearly, this matrix cannot be recovered from a subset of its
entries. Even if one sees 95\% of the entries sampled at random, then
we will miss elements in the first row with very high probability,
which makes the recovery of the vector $x$, and by extension of $M$,
impossible. The analogy in compressed sensing is that one obviously
cannot recover a signal assumed to be sparse in the time domain, by
subsampling in the time domain!

This example shows that one cannot hope to complete the matrix if some
of the singular vectors of the matrix are extremely sparse---above,
one cannot recover $M$ without sampling all the entries in the first
row, see \cite{CR08} for other related pathological examples. More
generally, if a row (or column) has no relationship to the other rows
(or columns) in the sense that it is approximately orthogonal, then
one would basically need to see all the entries in that row to recover
the matrix $M$. 
Such informal considerations led the authors of \cite{CR08} to
introduce a geometric incoherence assumption, 
% which quantifies the
% alignment between the canonical basis of $\R^{n_1}$ (resp.~$\R^{n_2}$)
% and the column (resp.~row) space of $M$. 
but for the moment, we will discuss an even simpler notion which
forces the singular vectors of $M$ to be spread across all
coordinates. To express this condition, recall the singular value
decomposition (SVD) of a matrix of rank $r$,
\begin{equation}
  \label{eq:svd}
  \mtx{M} = \sum_{k \in [r]} \sigma_k \vct{u}_k \vct{v}_k^*,
\end{equation}
in which $\sigma_1,\ldots,\sigma_r \geq 0$ are the singular values,
and $u_1,\ldots,u_r \in \R^{n_1}$, $v_1, \ldots, v_r \in \R^{n_2}$ are
the singular vectors. Our assumption is as follows: 
\begin{equation}
  \label{eq:bdd}
  \|u_k\|_{\ell_\infty} \le \sqrt{\mu_B/n_1}, \quad   \|v_k\|_{\ell_\infty} \le \sqrt{\mu_B/n}, 
\end{equation} 
for some $\mu_B \ge 1$, where the $\ell_\infty$ norm is of course
defined by $\|x\|_{\ell_\infty} = \max_i |x_i|$. We think of $\mu_B$
as being small, e.g. $O(1)$, so that the singular vectors are not too
spiky as explained above.

If the singular vectors of $M$ are sufficiently spread, the hope is
that there is a unique low-rank matrix which is consistent with the
observed entries. If this is the case, one could, in principle,
recover the unknown matrix by solving
\begin{equation}
\label{eq:rank}
  \begin{array}{ll}
    \text{minimize}   & \quad \rank(X)\\
    \text{subject to} & \quad \PO(\mtx{X}) = \PO(\mtx{M}), 
 \end{array}
\end{equation}
where $X \in \R^{n_1 \times n_2}$ is the decision variable.
Unfortunately, not only is this algorithm NP-hard, but all known
algorithms for exactly solving it are doubly exponential in theory and
in practice \cite{Grigoriev84}. This is analogous to the
intractability of $\ell_0$-minimization in sparse signal recovery.

A popular alternative is the convex relaxation
\cite{FazelThesis,fazelrank,Recht07,CR08,CT09}
\begin{equation}
\label{eq:cvx2}
  \begin{array}{ll}
    \text{minimize}   & \quad \|\mtx{X}\|_*\\
    \text{subject to} & \quad \PO(\mtx{X}) = \PO(\mtx{M}), 
 \end{array}
\end{equation}
(see \cite{Beck98,Mesbahi97} for the earlier related trace heuristic).
Just as $\ell_1$-minimization is the tightest convex relaxation of the
combinatorial $\ell_0$-minimization problem in the sense that the
$\ell_1$ ball of $\R^n$ is the convex hull of unit-normed 1-sparse
vectors (i.e.~vectors with at most one nonzero entry), nuclear-norm
minimization is the tightest convex relaxation of the NP-hard rank
minimization problem. To be sure, the nuclear ball $\{\mtx{X} \in
\R^{n_1 \times n_2} : \|\mtx{X}\|_* \le 1\}$ is the convex hull of the
set of rank-one matrices with spectral norm bounded by one. Moreover,
in compressed sensing, $\ell_1$ minimization subject to linear
equality constraints can be cast as a linear program (LP) for the
$\ell_1$ norm has an LP characterization: indeed for each $x \in
\R^n$, $\|x\|_{\ell_1}$ is the optimal value of
\begin{equation*}
  \begin{array}{ll}
    \text{maximize}   & \quad \<u,x\>\\
    \text{subject to} & \quad \|u\|_{\ell_\infty} \le 1, 
 \end{array}
\end{equation*}
with decision variable $u \in \R^n$. In the same vein, the nuclear
norm of $X \in \R^{n_1 \times n_2}$ has the SDP characterization
\begin{equation}
\label{eq:SDP}
 \begin{array}{ll}
   \text{maximize}   & \quad \<W,X\> \\
   \text{subject to} & \quad \|W\| \le 1,  
 \end{array}
\end{equation}
with decision variable $W \in \R^{n_1 \times n_2}$. This expresses the
fact that the spectral norm is dual to the nuclear norm. The
constraint on the spectral norm of $W$ is an SDP constraint since it
is equivalent to
\[
\begin{bmatrix} I_{n_1} & \mtx{W}\\
  \mtx{W}^* & \mtx{I}_{n_2} 
\end{bmatrix} \succeq 0, 
\]
where $I_{n}$ is the $n \times n$ identity matrix.  Hence,
\eqref{eq:cvx2} is an SDP, which one can express by writing $\|X\|_*$
as the optimal value of the SDP dual to \eqref{eq:SDP}.
%  the nuclear norm of a matrix has
% a semidefinite programming (SDP) characterization since one can check
% that $\|X\|_*$ is the optimal value of the SDP 
% \[
%  \begin{array}{ll}
%    \text{minimize}   & \quad \frac{1}{2} \bigl(\trace(\mtx{W}_1) + \trace(\mtx{W}_2\bigr)\\
%    \text{subject to} & \quad \begin{bmatrix} \mtx{W}_1 & \mtx{X}\\
%      \mtx{X}^* & \mtx{W}_2
% \end{bmatrix} \succeq 0
%  \end{array}
% \]
% with optimization variables $\mtx{W}_1 \in \R^{n_1\times n_1}$,
% $\mtx{W}_2 \in \R^{n_2\times n_2}$. Hence, \eqref{eq:cvx2} is an SDP.
 
In \cite{CT09}, it is proven that nuclear-norm minimization succeeds
nearly as soon as recovery is possible by any method whatsoever. 
\begin{theorem}\cite{CT09}
\label{teo:main1}
Let $M \in \R^{n_1 \times n_2}$ be a fixed matrix of rank $r = O(1)$
obeying \eqref{eq:bdd} and set $n := \max(n_1,n_2)$.  Suppose we
observe $m$ entries of $\mtx{M}$ with locations sampled uniformly at
random.  Then there is a positive numerical constant $C$ such that if
  \begin{equation}
    \label{eq:main1}
    m \ge C\, \mu_B^4\, n  \log^2 n,
  \end{equation}
  then $M$ is the unique solution to \eqref{eq:cvx2} with probability
  at least $1 - n^{-3}$.  In other words: with high probability,
  nuclear-norm minimization recovers all the entries of $\mtx{M}$ with
  no error.
\end{theorem}

As a side remark, one can obtain a probability of success at least
$1-n^{-\beta}$ for $\beta$ by taking $C$ in \eqref{eq:main1} of the
form $C' \beta$ for some universal constant $C'$.

An $n_1 \times n_2$ matrix of rank $r$ depends upon $r(n_1 + n_2 - r)$
degrees of freedom\footnote{This can be seen by counting the degrees
  of freedom in the singular value decomposition.}. When $r$ is small,
the number of degrees of freedom is much less than $n_1 n_2$ and this
is the reason why subsampling is possible. (In compressed sensing, the
number of degrees of freedom corresponds to the sparsity of the
signal; i.e.~the number of nonzero entries.) What is remarkable here,
is that exact recovery by nuclear norm minimization occurs as soon as
the sample size exceeds the number of degrees of freedom by a couple
of logarithmic factors.  Further, observe that if $\Omega$ completely
misses one of the rows (e.g.~one has no rating about one user) or one
of the columns (e.g.~one has no rating about one movie), then one
cannot hope to recover even a matrix of rank $1$ of the form $M = x
y^*$.  Thus one needs to sample every row (and also every column) of
the matrix.  When $\Omega$ is sampled at random, it is well
established that one needs at least on the order $O(n \log n)$ for
this to happen as this is the famous coupon collector's
problem. Hence, \eqref{eq:main1} misses the information theoretic
limit by at most a logarithmic factor.

To obtain similar results for all values of the rank, \cite{CT09}
introduces the {\em strong incoherence property} with parameter
$\mu$ stated below.
\begin{description}
\item[{A1}] Let $P_U$ (resp.~$P_V$) be the orthogonal projection onto
  the singular vectors $u_1, \ldots, u_r$ (resp.~$v_1, \ldots$, $v_r$).
  For all pairs $(a,a') \in [n_1] \times [n_1]$ and $(b,b') \in [n_2]
  \times [n_2]$,
\begin{align*}
\Bigl|\<e_a, P_U e_{a'}\> - \frac{r}{n_1} 1_{a = a'}\Bigr| & \le \mu \frac{\sqrt{r}}{n_1},\\
\Bigl|\<e_b, P_V e_{b'}\> - \frac{r}{n_2} 1_{b = b'}\Bigr| & \le
\mu \frac{\sqrt{r}}{n_2}.
\end{align*}
\item[{A2}] Let $E$ be the ``sign matrix'' defined by 
  \begin{equation}
    \label{eq:E}
    E = \sum_{k \in [r]} u_k v_k^*. 
  \end{equation}
  For all $(a,b) \in [n_1] \times [n_2]$,
\begin{equation*}\label{eab}
 |E_{ab}| \le \mu \frac{\sqrt{r}}{\sqrt{n_1 n_2}}. 
\end{equation*}
\end{description}

These conditions do not assume anything about the singular values. As
we will see, incoherent matrices with a small value of the strong
incoherence parameter $\mu$ can be recovered from a minimal set of
entries. Before we state this result, it is important to note that
many model matrices obey the strong incoherence property with a small
value of $\mu$.
\begin{itemize}
\item Suppose the singular vectors obey \eqref{eq:bdd} with $\mu_B =
  O(1)$ (which informally says that the singular vectors are not
  spiky), then with the exception of a very few peculiar matrices, $M$
  obeys the strong incoherence property with $\mu = O(\sqrt{\log n})$.

\item Assume that the column matrices $[u_1, \ldots, u_r]$ and $[v_1,
  \ldots, v_r]$ are independent random orthogonal matrices, then with
  high probability, $M$ obeys the strong incoherence property with
  $\mu = O(\sqrt{\log n})$, at least when $r \ge \log n$ as to avoid
  small samples effects.
\end{itemize}

The sampling result below is general, nonasymptotic and optimal up to
a few logarithmic factors.
\begin{theorem}
  \label{teo:main2} \cite{CT09} With the same notations as in Theorem
  \ref{teo:main1}, there is a numerical constant $C$ such that if
  \begin{equation}
    \label{eq:main2}
    m \ge C\, \mu^2\, nr \log^6 n,
  \end{equation}
  $M$ is the unique solution to \eqref{eq:cvx2} with probability at
  least $1 - n^{-3}$.
\end{theorem}

In other words, if a matrix is strongly incoherent and the cardinality
of the sampled set is about the number of degrees of freedom times a
few logarithmic factors, then nuclear-norm minimization is exact. This
improves on an earlier result of Cand\`es and Recht \cite{CR08} who
proved---under slightly different assumptions---that on the order of
$n^{6/5} r \log n$ samples were sufficient, at least for values of the
rank obeying $r \le n^{1/5}$.

We would like to point out a result of a broadly similar nature, but
with a completely different recovery algorithm and with a somewhat
different range of applicability, which was recently established by
Keshavan, Oh, and Montanari \cite{MontanariISIT}. Their conditions are
related to the incoherence property introduced in \cite{CR08}, and are
also satisfied by a number of reasonable random matrix models. There
is, however, another condition which states that the singular values
of the unknown matrix cannot be too large or too small (the ratio
between the top and lowest value must be bounded).  This algorithm 1)
trims each row and column with too few entries; i.e~replaces the
entries in those rows and columns by zero and 2) computes the SVD of
the trimmed matrix and truncate it as to only keep the top $r$
singular values (note that the value of $r$ is needed here). The
result is that under some suitable conditions discussed above, this
recovers a good approximation to the matrix $M$ provided that the
number of samples be on the order of $nr$. The recovery is not exact
but only approximate although the authors have announced that one
could add extra steps to the algorithm to provide an exact recovery if
one has more samples (on the order of $nr \log n$). At the time of
this writing, such results are not yet publicly available.

\subsection{Geometry and dual certificates}
\label{sec:T}

We cannot possibly rehash the proof of Theorem \ref{teo:main2} from
\cite{CT09} in this paper, or even explain the main technical steps,
because of space limitations. We will, however, detail sufficient and
almost necessary conditions for the low-rank matrix $M$ to be the
unique solution to the SDP \eqref{eq:cvx2}. This will be useful to
establish stability results.

\begin{figure}
 \centering
\includegraphics[scale=.4]{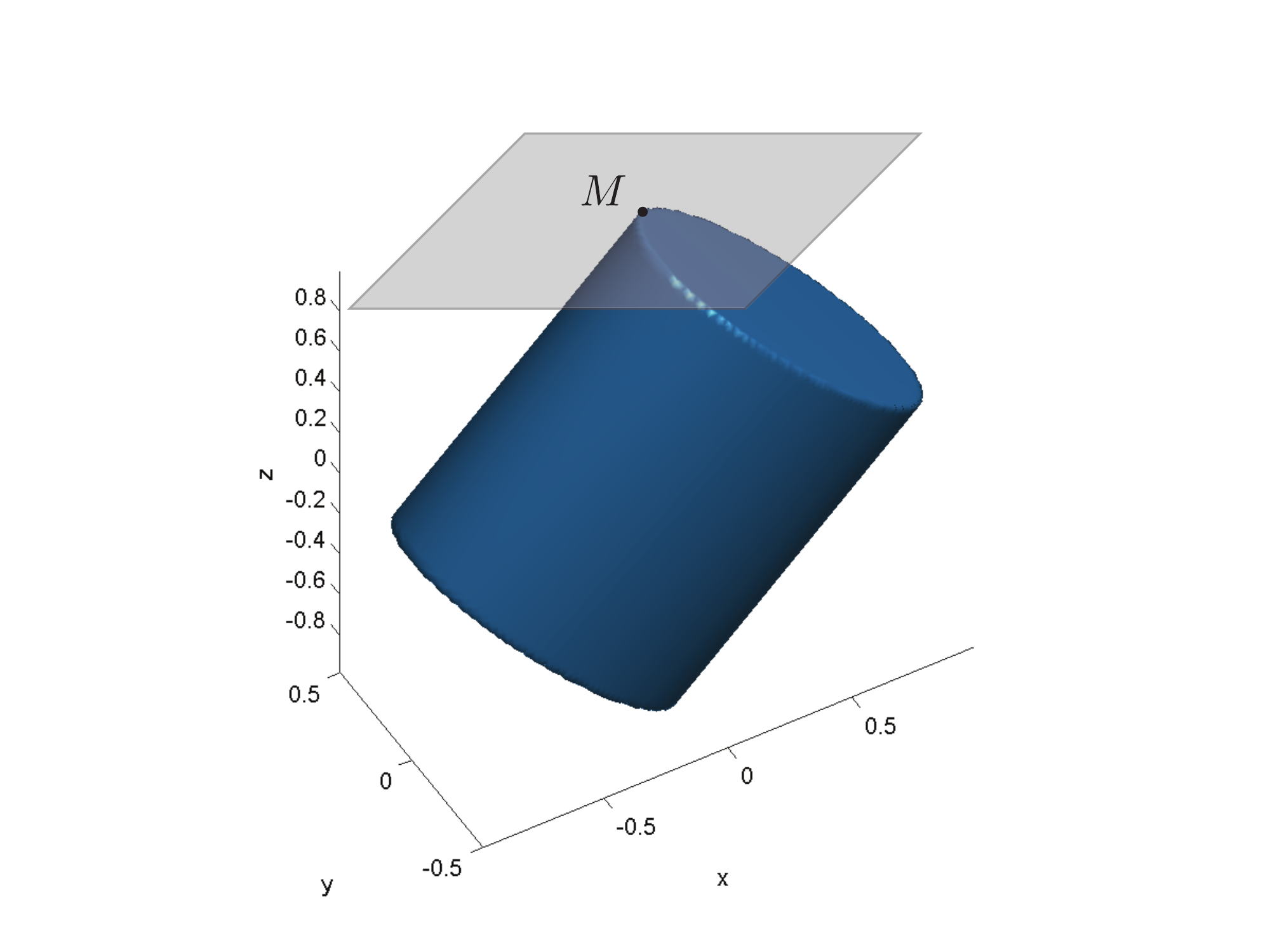}
\caption{\small The blue shape (courtesy of B.~Recht) represents the
  nuclear ball (see the main text), and the plane the feasible set
  (courtesy .}
\label{fig:nuclear_ball}
\end{figure}
The recovery is exact if the feasible set is tangent to the nuclear
ball at the point $M$, see Figure \ref{fig:nuclear_ball} which
represents the set of points $(x,y,z) \in \R^3$ such that the $2
\times 2$ {\em symmetric} matrix $\begin{bmatrix} x & y \\ y &
  z \end{bmatrix}$ has nuclear norm bounded by one.  To express this
mathematically\footnote{In general, $M$ minimizes the nuclear norm
  subject to the linear constraints $\mathcal{A}(X) = b$, $\mathcal{A}
  : \R^{n_1 \times n_2} \goto \R^{m}$, if and only if there is
  $\lambda \in \R^m$ such that $\mathcal{A}^*(\lambda) \in \partial
  \|M\|_*$.}, standard duality theory asserts that $M$ is a solution
to \eqref{eq:cvx2} if and only if there exists a dual matrix $\Lambda$
such that $\PO(\Lambda)$ is a subgradient of the nuclear norm at $M$,
written as
\begin{equation}
  \label{eq:dual}
  \PO(\Lambda) \in \partial \|M\|_*.
\end{equation}
Recall the SVD \eqref{eq:svd} of $M$ and the ``sign matrix'' $E$
\eqref{eq:E}. It is is well-known that $Z \in \partial \|M\|_*$ if and
only if $Z$ is of the form, 
\begin{equation}
  \label{eq:sub1}
  Z = E + W,
\end{equation}
where 
\begin{equation}
\label{eq:sub2}
P_U W = 0, \quad W P_V = 0, \quad \|W\| \le 1. 
\end{equation}
In English, $Z$ is a subgradient if it can be decomposed as the sign
matrix plus another matrix with spectral norm bounded by one, whose
column (resp.~row) space is orthogonal to the span of $u_1, \ldots,
u_r$, (resp.~of $v_1, \ldots, v_r$). Another way to put this is by
using notations introduced in \cite{CR08}. Let $T$ be the linear space
spanned by elements of the form ${u}_k {x}^*$ and ${y} {v}_k^*$, $k
\in [r]$, and let $T^\perp$ be the orthogonal complement to $T$. Note
that $T^\perp$ is the set of matrices obeying $P_U W = 0$ and $W P_V =
0$. Then, $Z \in \partial \|M\|_*$ if and only if 
\[
Z = E + \PTp(Z), \quad \|\PTp(Z)\| \le 1.  
\]
This motivates the following definition. 
\begin{definition}[Dual certificate] We say that $\Lambda$ is a dual
  certificate if $\Lambda$ is supported on $\Omega$ ($\Lambda =
  \PO(\Lambda)$), $\PT(\Lambda) = E$ and $\|\PTp(\Lambda)\| \le 1$.
\end{definition}

Before continuing, we would like to pause to observe the relationship
with $\ell_1$ minimization. The point $x^\star \in \R^n$ is solution
to
\begin{equation}
\label{eq:l1}
  \begin{array}{ll}
    \text{minimize}   & \quad \|x\|_{\ell_1}\\
    \text{subject to} & \quad Ax = b, 
 \end{array}
\end{equation}
with $A \in \R^{m \times n}$ if and only if there exists $\lambda \in
\R^m$ such that $ A^* \lambda \in \partial \|x^\star\|_{\ell_1}$.
Note that if $S^\star$ is the support of $x^\star$, $z \in \partial
\|x^\star\|_{\ell_1}$ is equivalent to
\[
z = e + w, \quad e = \begin{cases} \sgn(x^\star_i), & i \in S^\star,\\
  0, & i \notin  S^*, \end{cases} 
\]
and 
\[
w_i = 0 \text{ for all }  i \in S, \quad \|w\|_{\ell_\infty} \le 1.
\] 
Hence, there is a clear analogy and one can think of $T$ defined above
as playing the role of the support set in the sparse recovery problem.

With this in place, we shall make use of the following lemma from \cite{CR08}:
\begin{lemma}
  \label{teo:useful} \cite{CR08} Suppose there exists a dual
  certificate $\Lambda$ and consider any $H$ obeying $\PO(H) =
  0$. Then
\[
\|M + H\|_* \ge \|M\|_* - (1-\|\PTp(\Lambda)\|) \|\PTp(H)\|_*. 
\]
\end{lemma}

\begin{proof}
For any $Z \in \partial \|M\|_*$, we have 
\[
\|M+H\|_* \ge \|M\|_* + \<Z, H\>.
\]
With $\Lambda = E + \PTp(\Lambda)$ and $Z = E + \PTp(Z)$, we have 
\begin{align*}
\|M+H\|_* & \ge \|M\|_* + \<\Lambda,H\> + \<\PTp(Z-\Lambda),H\>\\
& = \|M\|_* + \<Z-\Lambda,\PTp(H)\>
\end{align*}
since $\PO(H) = 0$. Now we use the fact that the nuclear and spectral
norms are dual to one another. In particular, there exists $\|Z\| \le
1$ such that $\<Z,\PTp(H)\> = \|\PTp(H)\|_*$ and $|\<\Lambda,
\PTp(H)\>| = |\<\PTp(\Lambda),\PTp(H)\>| \le \|\PTp(\Lambda)\|
\|\PTp(H)\|_*$. Therefore,
\[
\|M+H\|_* \ge  \|M\|_* + (1-\|\PTp(\Lambda)\|) \|\PTp(H)\|_*, 
\]
which concludes the proof.
\end{proof}
A consequence of this lemma are the sufficient conditions below. 
\begin{lemma}
\label{teo:dual}
\cite{CR08} Suppose there exists a dual certificate obeying
$\|{\mathcal P}_{T^\perp}(\mtx{Y})\| < 1$ and that the restriction
$\PO\downharpoonright_T: T \to \PO(\R^{n \times n})$ of the (sampling)
operator $\PO$ restricted to $T$ is injective.  Then $\mtx{M}$ is the
unique solution to the convex program \eqref{eq:cvx2}.
\end{lemma}

\begin{proof}
  Consider any feasible perturbation $M + H$ obeying $\PO(H) =
  0$. Then by assumption, Lemma \ref{teo:useful} gives 
\[
\|M+H\|_* > \|M\|_*
\]
unless $\PTp(H) = 0$. Assume then that $\PTp(\mtx{H}) = 0$; that is to
say, $\mtx{H} \in T$. Then $\PO(\mtx{H}) = 0$ implies that $\mtx{H} =
0$ by the injectivity assumption. The conclusion is that $M$ is the
unique minimizer since any nontrivial perturbation increases the
nuclear norm.
\end{proof}

The methods for proving that matrix completion by nuclear minimization
is exact, consist in constructing a dual certificate.
\begin{theorem}
  \label{teo:certificate} 
  \cite{CT09} Under the assumptions of either Theorem \ref{teo:main1}
  or Theorem \ref{teo:main2}, there exists a dual certificate obeying
  $\|\PTp(\Lambda)\| \le 1/2$. In addition, if $p = m/(n_1 n_2)$ is
  the fraction of observed entries, the operator $\PT \PO \PT : T
  \goto T$ is one-to-one and obeys 
  \begin{equation}
    \label{eq:rudelson}
    \frac{p}{2} \, \OpId  \preceq \PT \PO \PT  \preceq  \frac{3p}{2} \, \OpId,  
  \end{equation}
  where $\OpId : T \goto T$ is the identity operator. 
\end{theorem}

The second part, namely, \eqref{eq:rudelson} shows that the mapping
$\PO : T \goto \R^{n_1 \times n_2}$ is injective. Hence, the
sufficient conditions of Lemma \ref{teo:dual} are verified, and the
recovery is exact. What is interesting, is that the existence of a
dual certificate together with the near-isometry
\eqref{eq:rudelson}---in fact, the lower bound---are sufficient to
establish the robustness of matrix completion vis a vis noise.

\section{Stable Matrix Completion}
\label{sec:stable}

In any real world application, one will only observe a few entries
corrupted at least by a small amount of noise. In the Netflix problem,
users' ratings are uncertain. In the system identification problem,
one cannot determine the locations $y(t)$ with infinite precision. In
the global positioning problem, local distances are imperfect. And
finally, in the remote sensing problem, the signal covariance matrix
is always modeled as being corrupted by the covariance of noise
signals. Hence, to be broadly applicable, we need to develop results
which guarantee that reasonably accurate matrix completion is possible
from noisy sampled entries. This section develops novel results
showing that this is, indeed, the case. 

Our noisy model assumes that we observe 
\begin{equation}
  \label{eq:noisy}
  Y_{ij} = M_{ij} + Z_{ij}, \quad (i,j) \in \Omega,  
\end{equation}
where $\{Z_{ij} : (i,j) \in \Omega\}$ is a noise term which may be
stochastic or deterministic (adversarial). Another way to express this
model is as 
\[
\PO(Y) = \PO(M) + \PO(Z), 
\]
where $Z$ is an $n \times n$ matrix with entries $Z_{ij}$ for $(i,j)
\in \Omega$ (note that the values of $Z$ outside of $\Omega$ are
irrelevant). All we assume is that $\|\PO(Z)\|_F \le \delta$ for some
$\delta > 0$. For example, if $\{Z_{ij}\}$ is a white noise sequence
with standard deviation $\sigma$, then $\delta^2 \le (m + \sqrt{8m})
\sigma^2$ with high probability, say.  To recover the unknown matrix,
we propose solving the following optimization problem:
\begin{equation}
  \label{eq:lasso}
  \begin{array}{ll}
    \textrm{minimize}   & \quad \|X\|_*\\
    \textrm{subject to} & \quad \|\PO(X - Y)\|_F \le \delta
 \end{array}
\end{equation} 
Among all matrices consistent with the data, find the one with minimum
nuclear norm. This is also an SDP, and let $\hat M$ be the solution to
this problem.

Our main result is that this reconstruction is accurate. 
\begin{theorem}
\label{teo:stable}
With the notations of Theorem \ref{teo:certificate}, suppose there
exists a dual certificate obeying $\|\PTp(\Lambda)\| \le 1/2$ and that
$\PT \PO \PT \succeq \frac{p}{2} \OpId$ (both these conditions are
true with very large probability under the assumptions of the
noiseless recovery Theorems \ref{teo:main1} and \ref{teo:main2}). Then
$\hat M$ obeys
  \begin{equation}
    \label{eq:stable}
    \|M - \hat M\|_F \le  4 \sqrt{\frac{C_p \min(n_1,n_2)}{p}} \, \delta + 2  \delta,
  \end{equation}
  with $C_p = 2 + p$.
\end{theorem}

For small values of $p$ (recall this is the fraction of observed
entries), the error is of course at most just about $
4\sqrt{\frac{2\min(n_1,n_2)}{p}}\, \delta$.  As we will see from the
proof, there is nothing special about 1/2 in the condition
$\|\PTp(\Lambda)\| \le 1/2$. All we need is that there is a dual
certificate obeying $\|\PTp(\Lambda)\| \le a$ for some $a < 1$ (the
value of $a$ only influences the numerical constant in
\eqref{eq:stable}). Further, when $Z$ is random, \eqref{eq:stable}
holds on the event $\|\PO(Z)\|_F \le \delta$.

Roughly speaking, our theorem states the following: {\em when perfect
  noiseless recovery occurs, then matrix completion is stable vis a
  vis perturbations}.  To be sure, the error is proportional to the
noise level $\delta$; when the noise level is small, the error is
small. Moreover, improving conditions under which noiseless recovery
occurs, has automatic consequences for the more realistic recovery
from noisy samples.

A significant novelty here is that there is just no equivalent of this
result in the compressed sensing or statistical literature for our
matrix completion problem does not obey the {\em restricted isometry
  property} (RIP) \cite{CT05}. For matrices, the RIP would
assume that the sampling operator obeys
\begin{equation}
\label{eq:rip}
(1-\delta) \|X\|_F^2 \le \frac{1}{p} \|\PO(X)\|_F^2 \le (1+\delta)
\|X\|_F^2
\end{equation}
for all matrices $X$ with sufficiently small rank and $\delta < 1$
sufficiently small \cite{Recht07}. However, the RIP does not hold
here. To see why, let the sampled set $\Omega$ be arbitrarily chosen
and fix $(i,j) \notin \Omega$. Then the rank-1 matrix $e_i e_j^*$
whose $(i,j)$th entry is 1, and vanishes everywhere else, obeys
$\PO(e_i e_j^*) = 0$. Clearly, this violates \eqref{eq:rip}.

It is nevertheless instructive to compare \eqref{eq:stable} with the
bound one would achieve if the RIP \eqref{eq:rip} were true. In this
case, \cite{FCRP08} would give
\[
\|\hat M - M\|_F \le C_0 p^{-1/2}\, \delta 
\]
for some numerical constant $C_0$. That is, an estimate which would be
better by a factor proportional to $1/\sqrt{\min(n_1,n_2)}$.  It would be
interesting to know whether or not estimates, which are as good as
what is achievable under the RIP, hold for the RIPless matrix
completion problem.  We will return to such comparisons later (Section
\ref{sec:oracle}).

%The second algorithm has worked very well in practice, and is more suited for situations in which the noise is spread out.  In this case, we are required to know that the maximum magnitude of an entry of the noise (on $\Omega$) is lower than some value, $\sigma$,
% \[
% \label{dantzigOptimization}
% \begin{array}{ll}
% \mbox{minimize} & \nucnorm{X}\\
% \mbox{subject to} & \max_{(i,j) \in \Omega} \abs{X_{ij} - Y_{ij}} \leq \sigma
% \end{array} 
% \]

We close this section by emphasizing that our methods are also applicable
to sparse signal recovery problems in which the RIP does not hold.

\subsection{Proof  of  Theorem \ref{teo:stable}}
\label{sec:proof}

\newcommand{\HO}{H_\Omega}
\newcommand{\HOc}{H_{\Omega^c}}

We use the notation of the previous section, and begin the proof by
observing two elementary properties. The first is that since $M$ is
feasible for \eqref{eq:lasso}, we have the {\em cone constraint}
\begin{equation}
  \label{eq:cone}
   \|\hat M\|_* \le \|M\|_*.  
\end{equation}
The second is that the triangle inequality implies the {\em tube
  constraint} 
\begin{align}
  \nonumber
  \|\PO(\hat M - M)\|_F & \le \|\PO(\hat M - Y)\|_F + \|\PO(Y - M)\|_F  \\
  \label{eq:tube}
& \le 2\delta,  
\end{align}
since $M$ is feasible. We will see that under our hypotheses,
\eqref{eq:cone} and \eqref{eq:tube} imply that $\hat M$ is close to
$M$.  Set $\hat M = M + H$ and put $\HO := \PO(H)$, $\HOc := \POc(H)$
for short. We need to bound $\|H\|^2_F = \|\HO\|_F^2 + \|\HOc\|_F^2$, 
and since \eqref{eq:tube} gives $\|\HO\|_F \le 2\delta$, it suffices to
bound $\|\HOc\|_F$. Note that by the Pythagorean identity, we have 
\begin{equation}
\label{eq:pythagoras}
\|\HOc\|_F^2 = \|\PT(\HOc)\|_F^2 + \|\PTp(\HOc)\|_F^2, 
\end{equation}
and it is thus sufficient to bound each term in the right hand-side.

We start with the second term. Let $\Lambda$ be a dual certificate obeying
$\|\PTp (\Lambda)\| \le 1/2$, we have
\[
 \|M+H\|_* \ge \|M+\HOc\|_* - \|H_\Omega\|_* 
\]
and 
\[
\|M+\HOc\|_* \ge  \|M\|_* + [1-\|\PTp (\Lambda)\|] \|\PTp (\HOc)\|_*. 
\]
The second inequality follows from Lemma \ref{teo:useful}. Therefore,
with $\|\PTp (\Lambda)\| \le 1/2$, the cone constraint gives
\[
\|M\|_* \ge \|M\|_* + \frac{1}{2} \|\PTp (\HOc)\|_* - \|\HO\|_*,
\]
or, equivalently,
\[
 \|\PTp (\HOc)\|_* \le  2 \|\HO\|_*. 
\]
Since the nuclear norm dominates the Frobenius norm, $\|\PTp
(\HOc)\|_F \le \|\PTp (\HOc)\|_*$, we have  
\begin{align}
\nonumber
 \|\PTp (\HOc)\|_F & \le 2 \|\HO\|_* \\
  \label{eq:PTpHOc}
  & \le 2\sqrt{n} \|\HO\|_F \le 4\sqrt{n} \delta,
\end{align}
where the second inequality follows from the Cauchy-Schwarz
inequality, and the last from \eqref{eq:tube}. 

To develop a bound on $\|\PT(\HO)\|_F$, observe that the assumption
$\PT \PO \PT \succeq \frac{p}{2} \OpId$ together with $\PT^2 = \PT$,
$\PO^2 = \PO$ give
\begin{align*}
  \|\PO \PT (\HOc)\|^2_F & = \<\PO \PT (\HOc), \PO \PT (\HOc)\> \\
  & = \<\PT \PO \PT(\HOc), \PT(\HOc)\> \\
  & \ge \frac{p}{2} \| \PT(\HOc)\|_F^2. 
\end{align*}
But since $\PO(\HOc) = 0 = \PO\PT(\HOc) + \PO\PTp(\HOc)$, we have
\begin{align*}
\|\PO \PT (\HOc)\|_F  & = \|\PO \PTp (\HOc)\|_F \\
& \le \|\PTp(\HOc)\|_F.  
\end{align*}
Hence, the last two inequalities give
\begin{equation}
\label{eq:PTHOc}
  \|\PT(\HOc)\|^2_F \le \frac{2}{p}  \|\PO \PT (\HOc)\|^2_F \le  \frac{2}{p}  \, \|\PTp(\HOc)\|^2_F. 
\end{equation}
As a consequence of this and \eqref{eq:pythagoras}, we have 
\begin{equation*}
  \label{eq:HOc}
  \|\HOc\|_F^2 \le \Bigl(\frac{2}{p}+1\Bigr)  \|\PTp(\HOc)\|_F^2. 
\end{equation*}
The theorem then follows from this inequality together with \eqref{eq:PTpHOc}.

\subsection{Comparison with an oracle}
\label{sec:oracle}

We would like to return to discussing the best possible accuracy one
could ever hope for.  For simplicity, assume that $n_1 = n_2 =
n$, and suppose that we have an oracle informing us about $T$. In many
ways, going back to the discussion from Section \ref{sec:T}, this is
analogous to giving away the support of the signal in compressed
sensing \cite{DS}. With this precious information, we would know that
$M$ lives in a linear space of dimension $2nr - r^2$ and would
probably solve the problem by the method of least squares:
\begin{equation}
  \label{eqn:LS}
  \begin{array}{ll}
    \textrm{minimize}   & \quad \|\PO(X) - \PO(Y)\|_F\\
    \textrm{subject to} & \quad X \in T.
 \end{array}
\end{equation}
That is, we would find the matrix in $T$, which best fits the data in a
least-squares sense. Let $\mathcal{A}: T \goto \Omega$ (we abuse
notations and let $\Omega$ be the range of $\PO$) defined by
$\mathcal{A} := \PO \PT$. Then assuming that the operator $\cA^*\cA =
\PT \PO \PT$ mapping $T$ onto $T$ is invertible (which is the case
under the hypotheses of Theorem \ref{teo:stable}), the least-squares
solution is given by
\begin{align}
\nonumber
M^{\text{Oracle}} & := (\cA^* \cA)^{-1} \cA^*(Y) \\
\label{eq:ls}
& = M + (\cA^* \cA)^{-1} \cA^*(Z).
\end{align}
Hence,
\[
\|M^{\text{Oracle}} - M\|_F = \|(\cA^* \cA)^{-1} \cA^*(Z)\|_F.
\]
Let $Z'$ be the minimal (normalized) eigenvector of $\cA^* \cA$ with
minimum eigenvalue $\lambda_{\text{min}}$, and set $Z = \delta
\lambda_{\text{min}}^{-1/2}\cA(Z')$ (note that by definition $\PO(Z) =
Z$ since $Z$ is in the range of $\cA$). By construction, $\|Z\|_F =
\delta$, and
\[
\|(\cA^* \cA)^{-1} \cA^*(Z)\|_F = \lambda_{\text{min}}^{-1/2} \delta
\gtrsim p^{-1/2} \, \delta
\]
since by assumption, all the eigenvalues of $\cA^*\cA = \PT \PO \PT$
lie in the interval $[p/2, 3p/2]$. The matrix $Z$ defined above also
maximizes $\|(\cA^* \cA)^{-1} \cA^*(Z)\|_F$ among all matrices bounded
by $\delta$ and so the oracle achieves 
\begin{equation}
  \label{eq:oracle1}
 \|M^{\text{Oracle}} - M\|_F \approx p^{-1/2} \delta
\end{equation}
with adversarial noise.  Consequently, our analysis looses a
$\sqrt{n}$ factor vis a vis an optimal bound that is achievable via
the help of an oracle.

The diligent reader may argue that the least-squares solution above
may not be of rank $r$ (it is at most of rank $2r$) and may thus argue
that this is not the strongest possible oracle.  However, as explained
below, if the oracle gave $T$ and $r$, then the best fit in $T$ of
rank $r$ would not do much better than \eqref{eq:oracle1}. In fact,
there is an elegant way to understand the significance of this oracle
which we now present. Consider a stronger oracle which reveals the row
space of the unknown matrix $M$ (and thus the rank of the
matrix). Then we would know that the unknown matrix is of the form
\[
M = M_C R^*, 
\]
where $M_C$ is an $n \times r$ matrix, and $R$ is an $n \times r$
matrix whose columns form an orthobasis for the row space (which we
can build since the oracle gave us perfect information). We would then
fit the $nr$ unknown entries by the method of least squares and find
$X \in \R^{n \times r}$ minimizing
\[
\|\PO(XR^*) - \PO(Y)\|_F. 
\]
Using our previous notations, the oracle gives away $T_0 \subset T$
where $T_0$ is the span of elements of the form $y v_k^*$, $k \in
[r]$, and is more precise. If $\mathcal{A}_0 : T_0 \goto \Omega$ is
defined by $\mathcal{A}_0 := \PO \mathcal{P}_{T_0}$, then the
least-squares solution is now
\[
(\cA_0^* \cA_0)^{-1} \cA_0^*(Y).
\]
Because all the eigenvalues of $\cA_0^* \cA_0$ belong to
$[\lambda_{\text{min}}(\cA^* \cA), \lambda_{\text{max}}(\cA^* \cA)]$,
the previous analysis applies and this stronger oracle would also
achieve an error of size about $p^{-1/2}\delta$. In conclusion, when
all we know is $\|\PO(Z)\|_F \le \delta$, one cannot hope for a
root-mean squared error better than $p^{-1/2} \delta$.

Note that when the noise is stochastic, e.~g.~when $Z_{ij}$ is white
noise with standard deviation $\sigma$, the oracle gives an error
bound which is adaptive, and is smaller as the rank gets
smaller. Indeed, $\E \|(\cA^* \cA)^{-1} \cA^*(Z)\|^2_F$ is equal to
\begin{equation}
 \sigma^2 \trace((\cA^* \cA)^{-1}) \approx \frac{2nr - r^2}{p} \sigma^2
\approx \frac{2nr}{p} \sigma^2 \label{eq:oracle2},
\end{equation}
since all the $2nr-r^2$ eigenvalues of $(\cA^* \cA)^{-1}$ are just
about equal to $p^{-1}$. When $nr \ll m$, this is better than
\eqref{eq:oracle1}.

\section{Numerical Experiments}
\label{sec:numerical}

We have seen that matrix completion is stable amid noise.  To
emphasize the practical nature of this result, a series of numerical
matrix completion experiments were run with noisy data.  To be
precise, for several values of the dimension $n$ (our first
experiments concern $n \times n$ matrices), the rank $r$, and the
fraction of observed entries $p = m/n^2$, the following numerical
simulations were repeated 20 times, and the errors averaged. A
rank-$r$ matrix $M$ is created as the product of two rectangular
matrices, $M=M_L M_R^*$, where the entries of $M_L, M_R \in \R^{n
  \times r}$ are iid $N(0, \sigma_n^2 := 20/\sqrt{n})$\footnote{The value of
  $\sigma_n$ is rather arbitrary.  Here, it is set so that the
  singular values of $M$ are quite larger than the singular values of
  $\PO(Z)$ so that $M$ can be distinguished from the null matrix.
  Having said that, note that for large $n$ and small $r$, the entries
  of $M$ are much smaller than those of the noise, and thus the signal
  appears to be completely buried in noise.}.  The sampled set
$\Omega$ is picked uniformly at random among all sets with $m$
entries.  The observations $\PO(Y)$ are corrupted by noise as in
\eqref{eq:noisy}, where $\{Z_{ij}\}$ is iid $N(0, \sigma^2)$; here, we
take $\sigma = 1$.  Lastly, $\hat{M}$ is recovered as the solution to
\eqref{eq:fpc} below.

For a peek at the results, consider Table \ref{tab:firstLook}.
\begin{table}
\begin{center}
\begin{tabular}{|c|cccc|}
\hline
$n$ & $100$ & $200$ & $500$ & $1000$\\
\hline
RMS error & $.99$ & $.61$ & $.34$ & $.24$\\
\hline
\end{tabular}
\end{center}
\caption{RMS error ($\|\hat M-M\|_F/n$) as a function of $n$ when subsampling 20\% of an $n \times n$ matrix of rank two. Each RMS error is averaged over 20 experiments.}
\label{tab:firstLook}
\end{table}
The RMS error defined as $\|\hat{M} - M\|_F/n$, measures the root-mean
squared error per entry.  From the table, one can see that even though
each entry is corrupted by noise with variance 1, when $M$ is a 1000
by 1000 matrix, the RMS error per entry is .24.  To see the
significance of this, suppose one had the chance to see {\em all} the
entries of the noisy matrix $Y = M + Z$.  Naively accepting $Y$ as an
estimate of $M$ would lead to an expected MS error of $\E \|Y -
M\|_F^2/n^2 = \E \|Z\|_F^2/n^2 = 1$, whereas the RMS error achieved
from only viewing 20\% of the entries is $\|\hat{M} - M\|_F^2/n^2 =
.24^2 = .0576$ when solving the SDP (\ref{eq:fpc})! Not only are we
guessing accurately the entries we have not seen, but we also
``denoise'' those we have seen.

In order to stably recover $M$ from a fraction of noisy entries,
the following regularized nuclear norm minimization problem was solved
using the FPC algorithm from \cite{Ma08}, 
\begin{equation}
\label{eq:fpc}
\textrm{minimize}   \quad 
\frac{1}{2} \|\PO(X - Y)\|_F^2 + \mu \|X\|_*.
\end{equation}
It is a standard duality result that \eqref{eq:fpc} is equivalent to
\eqref{eq:lasso}, for some value of $\mu$, and thus one could use
\eqref{eq:fpc} to solve \eqref{eq:lasso} by searching for the value of
$\mu(\delta)$ giving $\|\PO(\hat{M} - Y)\|_F = \delta$ (assuming
$\|\PO(Y)\|_F > \delta$).  We use \eqref{eq:fpc} because it works well
in practice, and because the FPC algorithm solves \eqref{eq:fpc}
nicely and accurately.
% yp: I just meant that the algorithm works much faster when \mu is
% large in comparison to the entries of $Y$, and \mu is chosen to be
% large when there is a low signal to noise ratio, but we can leave
% this out.
We also remark that a variation on our stability proof could also give
a stable error bound when using the SDP \eqref{eq:fpc}.

It is vital to choose a suitable value of $\mu$, which we do with the
following heuristic argument: first, simplifying to the case when
$\Omega$ is the set of all elements of the matrix, note that the
solution of \eqref{eq:fpc} is equal to $Y$ but with singular values
shifted towards zero by $\mu$ (soft-thresholding), as can be seen from
the optimality conditions of Section \ref{sec:exact} by means of
subgradients, or see \cite{Cai08}.  When $\Omega$ is not the entire
set, the solution is no longer exactly a soft-thresholded version of
$Y$, but experimentally, it is generally close.  Thus, we want to pick
$\mu$ large enough to threshold away the noise (keep the variance
low), and small enough not to overshrink the original matrix (keep the
bias low).  To this end, $\mu$ is set to be the smallest possible
value such that if $M = 0$ and $Y = Z$, then it is likely that the
minimizer of \eqref{eq:fpc} satisfies $\hat{M} = 0$.  It can be seen
that the solution to (\ref{eq:fpc}) is $\hat{M} = 0$ if $\|\PO(Y)\|
\leq \mu$ (once again, check the subgradient or \cite{Cai08}).  Then
the question is: what is $\|\PO(Z)\|$?  If we make a nonessential
change in the way $\Omega$ is sampled, then the answer follows from
random matrix theory.  Rather than picking $\Omega$ uniformly at
random, choose $\Omega$ by selecting each entry with probability $p$,
independently of the others.  With this modification, each entry of
$\PO(Z)$ is iid with variance $p \sigma^2$.  Then if $Z \in
\R^{n\times n}$, it is known that
% as shown in \cite{randomMatrix}
$n^{-1/2}\, \|\PO(Z)\| \rightarrow \sqrt{2p} \sigma$, almost surely as
$n \rightarrow \infty$.  Thus we pick $\mu = \sqrt{2np}\sigma$, where
$p = m/n^2$.  In practice, this value of $\mu$ seems to work very well
for square matrices. For $n_1 \times n_2$ matrices, based on the same
considerations, the proposal is $\mu = (\sqrt{n_1} +
\sqrt{n_2})\sqrt{p} \sigma$ with $p = m/(n_1n_2)$.

In order to interpret our numerical results, they are compared to
those achieved by the oracle, see Section \ref{sec:oracle}.  To this
end, Figure \ref{fig:errors} plots three curves for varying values of
$n, p,$ and $r$: 1) the RMS error introduced above, 2) the RMS error
achievable when the oracle reveals $T$, and the problem is solved
using least squares, 3) the estimated oracle root expected MS error
derived in Section \ref{sec:oracle}, i.e.~$\sqrt{\text{df}/[n^2p]} =
\sqrt{\text{df}/m}$, where $\text{df} = r(2n-r)$.  In our experiments,
as $n$ and $m/\text{df}$ increased, with $r=2$, the RMS error of the
nuclear norm problem appeared to be fit very well by $1.68\,
\sqrt{\text{df}/m}$.  Thus, to compare the oracle error to the actual
recovered error, we plotted the oracle errors times 1.68.  We also
note that in our experiments, the RMS error was never greater than
$2.25\, \sqrt{\text{df}/m}$.
\begin{figure}
\begin{center}
\begin{tabular}{c}
  \includegraphics[scale=0.464]{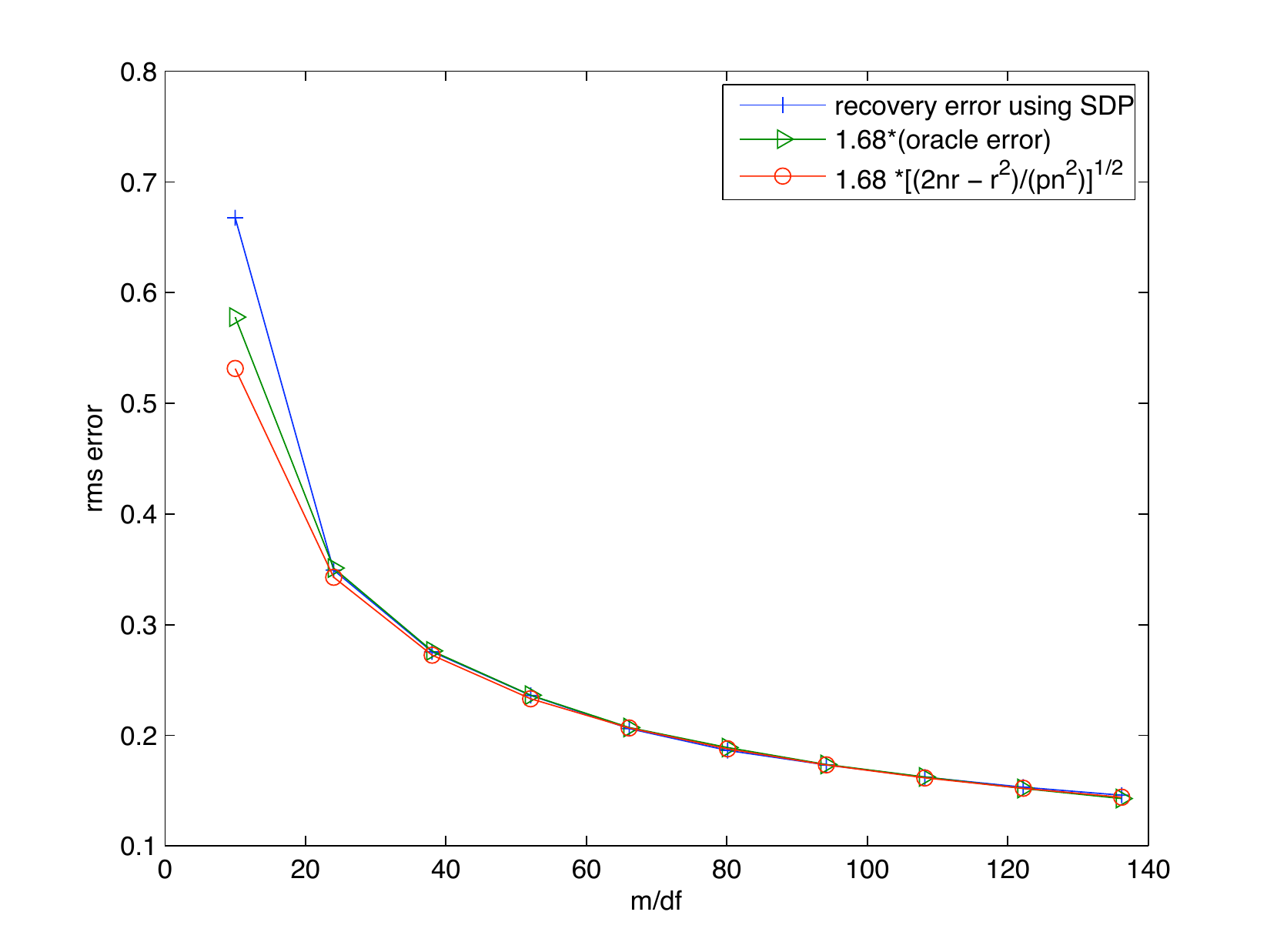}\\
  \includegraphics[scale=0.464]{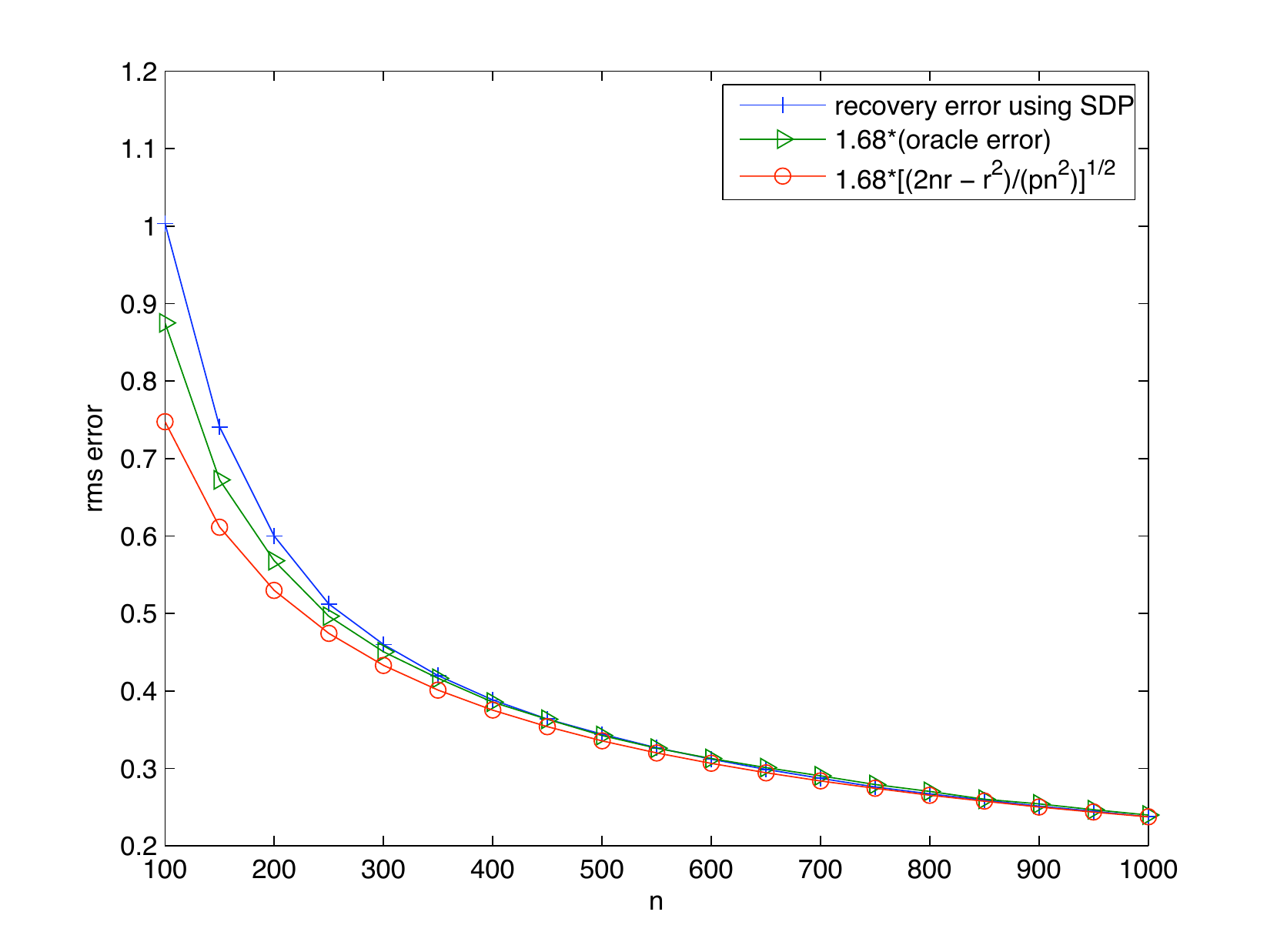}\\
  \includegraphics[scale=0.49]{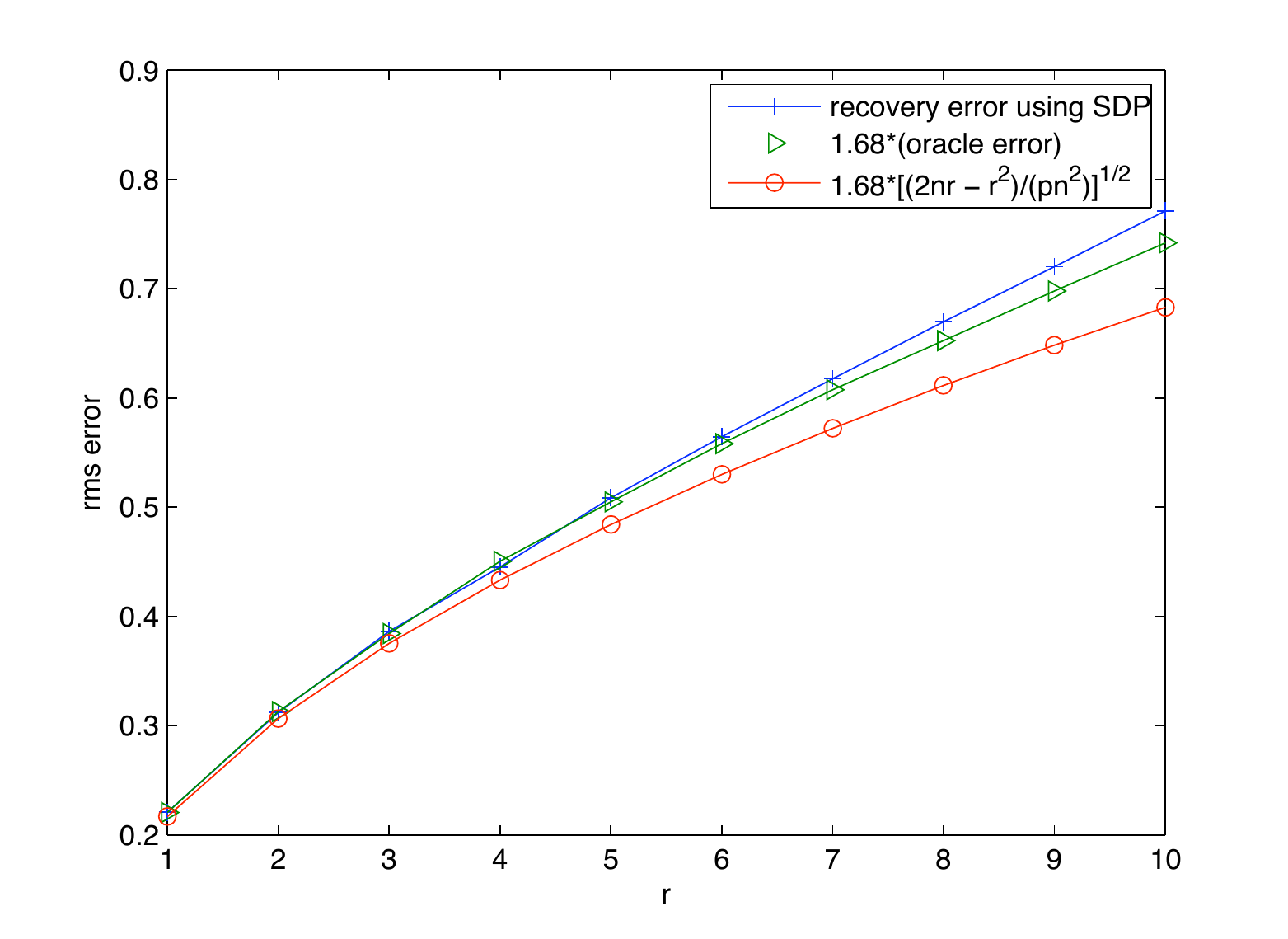}
\end{tabular}
\end{center}
\caption{\small Comparison between the recovery error, the oracle
  error times 1.68, and the estimated oracle error times 1.68.  Each
  point on the plot corresponds to an average over 20 trials. Top: in
  this experiment, $n = 600, r = 2$ and $p$ varies. The x-axis is the
  number of measurements per degree of freedom (df).  Middle: $n$
  varies whereas $r = 2$, $p = .2$.  Bottom: $n = 600$, $r$ varies and
  $p = .2$.}
\label{fig:errors}
\end{figure}

{\em No one can predict the weather.}  We conclude the numerical
section with a real world example.  We retrieved from the website
\cite{temp} a $366 \times 1472$ matrix whose entries are daily average
temperatures at 1472 different weather stations throughout the world
in 2008.  Checking its SVD reveals that this is an approximately low
rank matrix as expected.  In fact, letting $M$ be the temperature
matrix, and calling $M_2$ the matrix created by truncating the SVD
after the top two singular values gives $\|M_2\|_F/\|M\|_F = .9927$.

To test the performance of our matrix completion algorithm, we
subsampled 30\% of $M$ and then recovered an estimate, $\hat{M}$,
using (\ref{eq:fpc}).  Note that this is a much different problem than
those proposed earlier in this section.  Here, we attempt to recover a
matrix that is not exactly low rank, but only approximately. The
solution gives a relative error of $\|\hat{M} - M\|_F/\|M\|_F = .166$.
For comparison\footnote{The number 2 is somewhat arbitrary here,
  although we picked it because there is a large drop-off in the size
  of the singular values after the second.  If, for example, $M_{10}$
  is the best rank-10 approximation, then $\|M_{10} - M\|_F/\|M\|_F =
  .081$.}, exact knowledge of the best rank-2 approximation achieves
$\|M_2 - M\|_F/\|M\|_F = .121$.  Here $\mu$ has been selected to give
a good cross-validated error and is about 535.
% \red{I got the surprising large value
%   of .121 as follows:
% \begin{align*}
% \|M_2 - M\|_F^2/\|M\|_F^2 &= (\|M\|_F^2 - \|M_2\|_F^2)/\|M\|_F^2\\
% &= 1 - \|M_2\|_F^2/\|M\|_F^2\\
% &= 1 - .9927^2\\
% \Rightarrow \|M_2 - M\|_F/\|M\|_F &= \sqrt{1 - .9927^2} = .121
% \end{align*}
% Also, I chose $\mu = 17*(\sqrt{n_1} + \sqrt{n_2})\sqrt{p} \approx
% 535$.  I tested different values for $\mu$ until I found the best
% one.}

% EJC: What value of mu did you chooise for this? 

\section{Discussion}
\label{sec:discussion}

This paper reviewed and developed some new results about matrix
completion. By and large, low-rank matrix recovery is a field in
complete infancy abounding with interesting and open questions, and
if the recent avalanche of results in compressed sensing is any
indication, it is likely that this field will experience tremendous
growth in the next few years. 

At an information theoretic level, one would like to know whether one
can recover low-rank matrices from a few general linear functionals,
i.~e.~from $\cA(M) = b$, where $\cA$ is a linear map from $\R^{n_1
  \times n_2} \goto \R^m$. In this direction, we would like to single
out the original result of Recht, Fazel and Parrilo \cite{Recht07} who
showed---by leveraging the techniques and proofs from the compressed
sensing literature---that if each measurement is of the form $\<A_k,
X\>$, where $A_k$ is an independent array of iid Gaussian variables (a
la compressed sensing), then the nuclear norm heuristics recovers
rank-$r$ matrices from on the order of $nr \log n$ such randomized
measurements.

At a computational level, one would like to have available a suite of
efficient algorithms for minimizing the nuclear norm under convex
constraints and, in general, for finding low-rank matrices obeying
convex constraints. Algorithms with impressive performance in some
situations have already been proposed \cite{Cai08,Ma08} but the
computational challenges of solving problems with millions if not
billions of unknowns obviously still require much research.

% As researchers develop more theory and methods, they will find and
% discover more applications, which will fuel further developments,
% %yp: I changed the end of this sentence to make it less general.  Is what I wrote correct?
% \red{a cycle that we have already witnessed in compressed sensing.}
% %thereby creating a healthy research cycle.

\subsection*{Acknowledgements}
E.~C. is supported by ONR grants N00014-09-1-0469 and N00014-08-1-0749
and by the Waterman Award from NSF. E.~C. would like to thank Terence
Tao and Stephen Becker for some very helpful discussions.

\small
\bibliographystyle{plain}
\bibliography{NoisyCompletion}

\end{document}